\documentclass{article}

\usepackage{hyperref}
\usepackage{amsmath}
\usepackage{amsthm}
\usepackage{amssymb}
\usepackage{amsfonts}
\usepackage{array}
\usepackage{paralist}

\usepackage[margin=2.5cm]{geometry}

\begin{document}

\newtheorem{theorem}{Theorem}[section]
\newtheorem{lemma}[theorem]{Lemma}
\newtheorem{proposition}[theorem]{Proposition}
\newtheorem{corollary}[theorem]{Corollary}
\newtheorem{conj}[theorem]{Conjecture}
\newtheorem{remark}[theorem]{Remark}
\newtheorem{propos}[theorem]{Proposition}
\newtheorem{claim}[theorem]{Claim}
\newtheorem{definition}[theorem]{Definition}

\newcommand{\partdiff}[2]{\frac{\partial {#1}}{\partial {#2}}}
\newcommand{\secdiff}[2]{\frac{\partial^2 {#1}}{\partial {#2}^2}}
\newcommand{\mixdiff}[3]{\frac{\partial^2 {#1}}{{\partial {#2}}{\partial {#3}}}}

\def\E{{\bf E}}
\def\b1{{\bf 1}}
\def\be{{\bf e}}
\def\bp{{\bf p}}
\def\RR{{\mathbb R}}
\def\ZZ{{\mathbb Z}}
\def\cA{{\cal A}}
\def\cB{{\cal B}}
\def\cI{{\cal I}}
\def\cS{{\cal S}}
\def\cD{{\cal D}}

\title{\Large Online submodular welfare maximization:\\ Greedy is optimal}

\newcommand*\samethanks[1][\value{footnote}]{\footnotemark[#1]}

\author{Michael Kapralov\thanks{
CSAIL, Massachusetts Institute of Technology, Cambridge, MA. Work done while at Stanford University.
Email: \href{mailto:kapralov@mit.edu}{kapralov@mit.edu}. Research supported by NSF grant 0904325.}
\thanks{
We also acknowledge financial support from grant \#FA9550-12-1-0411
from the U.S. Air Force Office of Scientific Research (AFOSR) and the
Defense Advanced Research Projects Agency (DARPA).
}
\and Ian Post\thanks{Department of Combinatorics \& Optimization, University of Waterloo, Waterloo, ON, Canada.
Work done while at Stanford University.
Email: \href{mailto:ipost@uwaterloo.ca}{ipost@uwaterloo.ca}. Research supported by NSF grant 0915040.}
\samethanks[2]
\and Jan Vondr\'ak\thanks{
IBM Almaden Research Center, San Jose, CA. E-mail: \href{mailto:jvondrak@us.ibm.com}{jvondrak@us.ibm.com}.}
}

\date{}

\maketitle

\abstract{We prove that no online algorithm (even randomized, against an oblivious adversary) is better than $1/2$-competitive for welfare maximization with coverage valuations, unless $NP = RP$. Since the Greedy algorithm is known to be $1/2$-competitive for monotone submodular valuations, of which coverage is a special case, this proves that Greedy provides the optimal competitive ratio. On the other hand, we prove that Greedy in a stochastic setting with i.i.d.~items and valuations satisfying diminishing returns is $(1-1/e)$-competitive, which is optimal even for coverage valuations, unless $NP=RP$. For online budget-additive allocation, we prove that no algorithm can be $0.612$-competitive with respect to a natural LP which has been used previously for this problem.}

\section{Introduction}

We study an online variant of the welfare maximization problem in combinatorial auctions: $m$ items are arriving online, and each item should be allocated upon arrival to one of $n$ agents whose interest in different subsets of items is expressed by valuation functions $w_i: 2^{[m]} \rightarrow \RR_+$. The goal is to maximize $\sum_{i=1}^{n} w_i(S_i)$ where $S_i$ is the set of items allocated to agent $i$. Variants of the problem arise by considering different classes of valuation functions $w_i$ and different models (adversarial/stochastic) for the arrival ordering of the items. We remark that in this work we do not consider any game-theoretic aspects of this problem.

The origin of this line of work can be traced back to a seminal paper of Karp, Vazirani and Vazirani \cite{KVV90} on online bipartite matching. This can be viewed as a welfare maximization problem where one side of the bipartite graph represents agents and the other side items; each agent $i$ is interested in the items $N(i)$ joined to $i$ by an edge, and he is completely satisfied by 1 item, meaning the valuation function can be written as $w_i(S) = \min \{ |S \cap N(i)|, 1 \}$. Karp, Vazirani and Vazirani gave an elegant $(1-1/e)$-competitive randomized algorithm, which improves a greedy $1/2$-approximation and is optimal in this setting.

Recent interest in online allocation problems arises from applications in online advertising, where the items represent ad slots associated with search queries, and agents are advertisers interested in having their ad displayed in connection with certain queries. A popular model in this context is the {\em budget-additive} framework \cite{GKP01,MSVV05,BJN07} where valuations have the form $w_i(S) = \min \{ \sum_{j \in S} b_{ij}, B_i \}$. More generally, combinatorial auctions \cite{LLN06} form a setting where multiple items are sold to multiple agents with valuation functions $w_i$. Again, practical restrictions often require that the decision about each item needs to be made immediately, rather than after seeing the entire pool of items. Hence the online model, which we study in this paper.

The baseline algorithm in this setting is the {\em greedy algorithm}, due to Fisher, Nemhauser and Wolsey, who initiated the study of problems involving maximization of submodular functions \cite{NWF78, FNW78, NW78}. The greedy algorithm simply allocates each incoming item to the agent who gains from it the most and is $1/2$-competitive whenever the valuation functions of the agents are {\em monotone submodular} \cite{FNW78,LLN06}. This is in fact the most general setting known where a constant-factor approximation can be achieved even for the offline welfare maximization problem (using {\em value queries}; more general classes of valuations can be handled when more powerful queries are available \cite{Fei06}). Thus the basic question in most variants of this problem is whether the factor of $1/2$ is optimal or can be improved. For the offline welfare maximization problem with monotone submodular valuations, a $(1-1/e)$-approximation has been found \cite{V08}, and this is optimal \cite{KLMM08}.

At the other end of the spectrum is the above-mentioned bipartite matching problem, which can be viewed as a welfare maximization problem with valuation functions of the form $w_i(S) = \min \{ |S \cap N(i)|, 1\}$ (a very special case of a submodular function). The $(1-1/e)$-competitive algorithm of \cite{KVV90} is optimal in the adversarial online setting; several improvements have been obtained in various stochastic settings \cite{GM08,FMMM09,BK10,MOS11,MOZ12}. Factor $(1-1/e)$-competitive algorithms have been also found in two adversarial budget-additive settings, the small-bids case, $w_i(S) = \min \{ \sum_{j \in S} b_{ij}, B_i\}$ where $b_{ij} \ll B_i$ \cite{MSVV05}, and the single-bids case, $w_i(S) = \min \{ \sum_{j \in S} b_{ij}, B_i\}$ where $b_{ij} \in \{0,b_i\}$ for some $b_i$ independent of $j$ \cite{AGKM11}. A unifying generalization of these $(1-1/e)$-competitive algorithms to the budget-additive setting, $w_i(S) = \min \{ \sum_{j \in S} b_{ij}, B_i\}$, has been conjectured but still remains open. Prior to this work, it was conceivable that a $(1-1/e)$-competitive algorithm might exist for arbitrary monotone submodular valuations, but the best known online algorithm gave only an $o(1)$ improvement over $1/2$ \cite{DS06}.

\paragraph{Our results.}
We prove that:
\begin{itemize}
\item In the online setting with submodular valuations, the factor of $1/2$ cannot be improved unless $NP=RP$ (even by randomized algorithms against an oblivious adversary). Hence, the greedy $1/2$-competitive algorithm is optimal up to lower-order terms. This holds in fact for the special case of {\em coverage valuations} (see Section~\ref{sec:online-coverage}).

\item In the online setting with budget-additive valuations, we prove that no (randomized) algorithm is $0.612$-competitive with respect to a natural LP, which has been used successfully in the special case of small bids \cite{BJN07}. Thus, a $(1-1/e)$-competitive algorithm would need to use a different approach (see Section~\ref{sec:budget-additive}).

\item In a stochastic setting with items arriving i.i.d.~from an unknown distribution, the greedy algorithm is $(1-1/e)$-competitive for valuations with the property of diminishing returns (a natural extension of submodularity to multisets which we define in Section~\ref{sec:diminishing}). This is optimal even for coverage valuations and a known (uniform) distribution, unless $NP=RP$ (see Section~\ref{sec:stochastic}).
\end{itemize}

\paragraph{Our techniques.}
Our hardness result for online algorithms in the adversarial setting relies on a combination of two sources of hardness: (1) the inapproximability of Max $k$-cover due to Feige \cite{Feige98}, and (2) the lack of information arising from the unknown online ordering. A careful combination of these two ingredients gives an optimal hardness result ($1/2+\epsilon$) for online algorithms under coverage valuations. Our hardness result in the i.i.d.~stochastic setting also relies on the hardness of Max $k$-cover. A consequence of our use of the computational hardness of (offline) Max $k$-cover is that our results rely on a complexity-theoretic assumption ($NP=RP$), which is somewhat unusual in the context of online algorithms.

Our negative result for budget-additive valuations is based on an integrality-gap example for the natural LP and does not rely on any complexity-theoretic assumption. This result does not rule out, e.g., a $(1-1/e)$-competitive algorithm for budget-additive valuations, but we consider it instructive, considering recent efforts to develop online algorithms in the primal-dual framework. Our result points to the fact that perhaps the natural LP is too weak for the general budget-additive setting, and stronger LPs such as the Configuration LP should be considered (also, see a discussion in \cite{CG08}).

Finally, our positive result for the greedy algorithm in the i.i.d.~stochastic setting is an extension of a similar analysis for budget-additive valuations \cite{DJSW11}. Here, we want to point out the definition of valuation functions satisfying the property of {\em diminishing returns} (Section~\ref{sec:diminishing}). This is a generalization of submodularity to functions on multisets. We remark that for set functions, the properties of diminishing returns and submodularity coincide, but this is not the case for functions on multisets. We believe that our generalization is a natural one, considering the original motivation for submodularity in the context of combinatorial auctions, and we wish to highlight this definition for possible future work.

\

In the following, we state our results more formally and present the proofs.

\section{Online allocation with coverage valuations}
\label{sec:online-coverage}

\paragraph{Welfare maximization.}
In the welfare maximization problem (sometimes also referred to as the ``allocation problem" or ``combinatorial auctions"), the goal is to allocate
$|M| = m$ items to $n$ agents with valuation functions $w_i:2^M \rightarrow \RR_+$ in a way that maximizes $\sum_{i=1}^{n} w_i(S_i)$,
where $S_i$ is the set of items allocated to agent $i$ (satisfying $S_i \cap S_j = \emptyset$ for $i \neq j$).

\paragraph{Online welfare maximization.}
In the online version of the problem, items arrive one by one and we have to allocate each item when it arrives,
knowing only the agents' valuations on the items that have arrived so far. An algorithm is $c$-competitive
if, for any ordering of the incoming items, it achieves at least a $c$-fraction of the (offline) optimal welfare.
A randomized algorithm is $c$-competitive against an oblivious adversary if, for any ordering of the incoming items
(fixed before running the algorithm), it achieves at least a $c$-fraction of the optimal welfare in expectation.

\paragraph{Coverage valuations.}
A valuation function $w:2^M \rightarrow \RR_+$ is called a {\em coverage valuation} if there is a set system $\{A_j: j \in M\}$ such that $w(S) = |\bigcup_{j \in S} A_j|$ for all $S \subseteq M$.

\paragraph{Submodular valuations.}
A valuation function $w:2^M \rightarrow \RR_+$ is called {\em submodular} if $w(S \cup T) + w(S \cap T) \leq w(S) + w(T)$.
It is called {\em monotone} if $w(S) \leq w(T)$ whenever $S \subseteq T$.

\paragraph{Succinct representation and oracles.}
For complexity-theoretic considerations, it is important how the valuation functions are presented on the input. In this paper, we assume that coverage valuations are presented explicitly, by a succinct representation of size polynomial in $|M|$. Submodular valuations are presented by means of a {\em value oracle}, which can answer queries in the form ``What is the value of $w_i(S)$?"

\

Our main result for online welfare maximization is as follows.

\begin{theorem}
\label{thm:online-coverage}
Unless $NP = RP$, there is no $(1/2+\delta)$-competitive polynomial-time algorithm (even randomized, against an oblivious adversary) for the online welfare maximization problem with coverage valuations and constant $\delta>0$.
\end{theorem}

Our main tool is Feige's hardness reduction, which proves the optimality of $(1-1/e)$-approximation for Max $k$-cover \cite{Feige98}. We also require some additional properties of this reduction, which have been described in \cite{FT04,FV10}. We summarize the properties that we need as follows:

\paragraph{Hardness of Max $k$-cover.} {\em
For any fixed $c_0>0$ and $\epsilon>0$, it is NP-hard to distinguish between the following two cases for a given collection of sets $\cS \subset 2^U$, partitioned into groups $\cS_1,\ldots,\cS_k$:
\begin{enumerate}
\item YES case: There are $k$ disjoint sets, $1$ from each group $\cS_i$, whose union is the universe $U$.
\item NO case: For any choice of $\ell \leq c_0 k$ sets, their union covers at most a $(1 - (1-1/k)^\ell + \epsilon)$-fraction of the elements of $U$.
\end{enumerate}
This holds even for set systems with the following properties:
\begin{itemize}
\item every set has the same (constant) size $s$; and
\item each group contains the same (constant) number of sets $n$.
\end{itemize}
}

As we show in Section~\ref{sec:offline}, this reduction also gives hard instances for welfare maximization with coverage valuations,
proving that any (offline) $(1-1/e+\delta)$-approximation would imply $P=NP$. In Section~\ref{sec:online}, we prove our result, Theorem~\ref{thm:online-coverage}.

\paragraph{Our approach.}
We produce instances of online welfare maximization by taking multiple copies of a hard instance $\cI$ for offline welfare maximization
and repeating them with certain (random) agents gradually dropping out of the system. We prove that an online algorithm faces two obstacles: it cannot solve the offline instance $\cI$ optimally (in fact it already loses a factor of $1-1/e$ there), and in addition it does not know in advance which agents will drop out at what time. A careful analysis of these two obstacles in combination gives the optimal hardness of $(1/2+\delta)$-approximation for online algorithms.

\subsection{Warm-up: hardness of offline welfare maximization}
\label{sec:offline}

First, we show how Feige's reduction implies the hardness $(1-1/e+\delta)$-approximation for welfare maximization with coverage valuations. This was previously proved by a more involved technique in \cite{KLMM08}. The result of \cite{KLMM08} has the additional property that it holds even when all agents have the same coverage valuation; in our reduction the valuations are different.

\paragraph{Reduction.} Consider a set system that forms a hard instance of Max $k$-cover as described above. We produce an instance of welfare maximization with $n$ agents and $m = kn$ items (where $n$ is the number of sets in each group, and $k$ is the number of groups). Each agent will have a valuation associated with this set system. However, the way items are associated with sets will be different for each agent. Let the $kn$ items be described by pairs $(j_1,j_2) \in [k] \times [n]$, and let the sets in the set system be denoted by $A_{j_1,j_2}$ where $(j_1,j_2) \in [k] \times [n]$. Then, the item $(j_1,j_2)$ for agent $i$ is associated with the set $A_{j_1,j_2+i \pmod n}$.
In other words, the value of a set of items $S$ for agent $i$ is
$$ w_i(S) = \Big| \bigcup_{(j_1,j_2) \in S} A_{j_1,{(j_2+i \bmod n)}} \Big|.$$

Now consider the two cases:
\begin{itemize}
\item YES case: There are $k$ sets, one from each group, covering the universe. Denote these sets by $A_{j,\pi(j)}$ for some function $\pi:[k] \rightarrow [n]$. Then, there is an allocation where agent $i$ receives the set of items $S_i = \{ (j, \pi(j)-i \bmod n): j \in [k] \}$.
Note that these sets of items are disjoint, due to the cyclic shift depending on $i$. Also, each agent is perfectly satisfied, since the union of the sets associated with her items is $\bigcup_{j \in [k]} A_{j, \pi(j)} = U$. Hence $w_i(S_i) = |U|$ for all $i$.
\item NO case: For each choice of $\ell \leq c_0 k$ sets, they cover at most a $(1 - (1-1/k)^\ell + \epsilon)$-fraction of the universe.
(We choose $c_0$ to be a large constant.) In other words, any agent who receives $\ell \leq c_0 k$ sets
gets value at most $w_i(S_i) \leq (1 - (1-1/k)^\ell + \epsilon) |U|$. Here, it does not matter who receives which items, as we have a bound depending solely on the number of items received. Since this bound is a concave function, the best possible welfare is achieved when each agent receives exactly $k$ items, and this yields welfare $(1 - (1-1/k)^k + \epsilon) |U|$ per agent. By choosing $k$ arbitrarily large and $\epsilon>0$ arbitrarily small, we obtain welfare arbitrarily close to $(1-1/e) |U|$ per agent.
\end{itemize}

In the following, we will use this hard instance of welfare maximization with coverage valuation as a black box.

\subsection{Hardness of online welfare maximization}
\label{sec:online}

Here we prove Theorem~\ref{thm:online-coverage}. We produce a reduction from Max $k$-cover to online welfare maximization as follows.

\paragraph{The hard online instances.}
Let $\cI$ be an instance of welfare maximization with coverage valuations, obtained from a hard instance of Max $k$-cover (as in Section~\ref{sec:offline}), with $n$ agents and $m = kn$ items. For a parameter $t \geq 1$, we produce the following instance $\cI^{(t)}$ of online welfare maximization, with $tn$ agents and $tm$ items, proceeding in $t$ stages:
\begin{itemize}
\item In the first stage, we have $t$ copies of each agent of the instance $\cI$, with exactly the same valuation function.
The valuation function for each agent is determined by the set system of $\cI$. The $m$ items of instance $\cI$ arrive in an arbitrary order.
\item After each stage, one copy of each agent is effectively ``deactivated", in the sense that all subsequent items have zero value for her.
The copy of each agent that disappears is chosen by an adversary.
\item In stage $t' \in \{1,\ldots,t\}$, we have $(t-t'+1)n$ ``active agents" remaining, who are still interested in the remaining items.
In each stage, $m$ items of the original instance arrive in an arbitrary order, but now they are valuable only for the remaining active agents. For these agents, the items are effectively copies of the items that arrived in previous stages, and they are represented by the same sets.
\end{itemize}

These instances were inspired by the $1-1/e$ lower bound for online matching \cite{KVV90}. Essentially, we take an instance of bipartite matching that is hard for online algorithms and expand each incoming vertex into an entire instance of welfare maximization with coverage valuations, to impose the additional difficulty of approximating an APX-hard problem at each stage.

We analyze this instance in a series of claims.

\begin{claim}
\label{cl:opt}
The offline optimum in the YES case is $t n |U|$.
\end{claim}

\begin{proof}
The offline optimum allocates all items in stage $j$ to those agents who will be deactivated at the end of this stage.
Since these are $n$ agents whose valuations of the items of this stage correspond exactly to the instance $\cI$,
in the YES case they can obtain optimal value $n |U|$ (since every agent can cover the universe $U$).
Adding up over all stages, the total value collected by all agents is $t n |U|$.
\end{proof}

\begin{claim}
\label{cl:bound}
Let the adversary choose a copy of each agent to be deactivated after each stage independently and uniformly at random from the remaining active copies. Then the expected total number of items allocated to the agents deactivated at the end of stage $j$ is at most $m \ln \frac{t}{t-j}$.
\end{claim}

\begin{proof}
Let $A_j$ denote the agents deactivated right after stage $j$; $A_j$ contains exactly 1 copy of each agent. Consider $i \leq j$ and condition on the set of agents active in stage $i$. The choice of which agents will appear in $A_j$ will be made after stage $j$, independently of what the algorithm does in stage $i$. Since the choice of $A_j$ is uniform in each stage $j$, each of the $t-i+1$ copies of a given agent active in stage $i$ has the same probability $(\frac{1}{t-i+1})$ of appearing in $A_j$.
The number of items allocated in each stage is $m$, hence the expected number of items allocated to $A_j$ in stage $i$ is $\frac{m}{t-i+1}$. By linearity of expectation, the number of items allocated to $A_j$ between stages $1$ and $j$ is
\begin{equation*}
 \sum_{i=1}^{j} \frac{m}{t-i+1} \leq \int_0^j \frac{m}{t-x} dx = m \ln \frac{t}{t-j}. \qedhere
\end{equation*}
\end{proof}

\begin{claim}
\label{cl:stage-j}
For every $\epsilon'>0$, there are $\epsilon,c_0>0$ and a constant lower bound on $k$ (parameters of the Max $k$-cover reduction) such that for every $j \leq (1-\epsilon') t$, the expected value collected in the NO case by the agents deactivated at the end of stage $j$ is at most $(j/t + \epsilon') n |U|$.
\end{claim}

\begin{proof}
Denote again by $A_j$ the agents deactivated at the end of stage $j$.
By Claim~\ref{cl:bound}, the expected number of items allocated to $A_j$ is at most $m \ln \frac{t}{t-j}$.
Let $\mu$ denote the expected number of items allocated per agent in $A_j$: 
we get $\mu \leq \frac{m}{n} \ln \frac{t}{t-j} = k \ln \frac{t}{t-j}$.
Assuming that $j \leq (1-\epsilon') t$, we have $\mu \leq k \ln \frac{1}{\epsilon'}$. Let us set  $c_0 = 1 + \ln \frac{1}{\epsilon'}$,
and let us denote by $\nu(\ell)$ the largest value that an agent can possibly obtain from $\ell$ items.
By properties of the NO case, we know that for $\ell \leq c_0 k$, we have $\nu(\ell) \leq (1 - (1-1/k)^\ell + \epsilon) |U|$. For $\epsilon = \frac12 \epsilon'$ and $k$ lower-bounded by some sufficiently large constant, we can replace this bound by $\nu(\ell) \leq (1 - e^{-\ell / k} + \epsilon') |U|$.

A technical point here is that this bound holds only for $\ell \leq c_0 k$, while the actual number of allocated items is random and could be much larger. However, we can deal with this issue as follows. Let us define $\phi(x) = (1 - e^{-x / k} + \epsilon') |U|$. The derivative of $\phi$ at $\mu$ is $\phi'(\mu) = \frac{1}{k} e^{-\mu/k} |U|$. Therefore, since the function $\phi(x)$ is concave, we have $\phi(x) \leq \phi(\mu) + \phi'(\mu) (x-\mu)$, for $\ell \in [0, c_0 k]$. Thus we obtain a (weaker) linear bound: 
\begin{equation*}
\nu(\ell) \leq \phi(\mu) + \phi'(\mu) (\ell-\mu)
= (1 - e^{-\mu/k} + \epsilon' + \frac{1}{k} e^{-\mu/k} (\ell-\mu)) |U|.
\end{equation*} 
Furthermore, we always have the trivial bound $\nu(\ell) \leq |U|$, for any $\ell$. This bound is anyway stronger than the one above for $\ell \geq c_0 k$, because we have $c_0 k \geq \mu + k$. Therefore, we obtain the following bound for all $\ell \geq 0$:
$$ \nu(\ell) \leq \min \{ |U|, (1 - e^{-\mu/k} + \epsilon' + \frac{1}{k} e^{-\mu/k} (\ell - \mu)) |U|\} ,$$
and this (piecewise linear) bound is still concave. Since the expected number of items per player is $\E[\ell] = \mu$, the worst case is that each agent in $A_j$ indeed receives $\mu$ items (deterministically), and her value is $\nu(\mu) \leq (1 - e^{-\mu/k} + \epsilon') |U|$.
Using our bound $\mu \leq k \ln \frac{t}{t-j}$, we obtain that the expected value
collected per agent in $A_j$ is at most $(1 - \frac{t-j}{t} + \epsilon') |U| = (\frac{j}{t} + \epsilon') |U|$.
\end{proof}

\begin{proof}[Theorem~\ref{thm:online-coverage}]
Let us assume now that there is a $(\frac12 + \delta)$-competitive algorithm for online welfare maximization with coverage valuations. We set $\epsilon' = \delta / 4$ and the parameters $c_0, \epsilon$ accordingly to this value of $\epsilon'$ (see Claim~\ref{cl:stage-j}).  Given an instance $\cI$ of Max $k$-cover, we can also assume that $k$ is sufficiently large as required by Claim~\ref{cl:stage-j};
otherwise all parameters of the Max $k$-cover instance are constant, and we can solve it by exhaustive search. If $k$ is sufficiently large, we run the presumed online algorithm on the random instance $\cI^{(t)}$ that we constructed above.

In the NO case, denote by $V_j$ the expected value collected by agents deactivated after stage $j$. By Claim~\ref{cl:stage-j}, we have $V_j \leq (\frac{j}{t} + \epsilon') n |U|$ for $j \leq (1-\epsilon') t$. The value collected by the agents deactivated in each of the last $\epsilon' t$ stages is $V_j \leq n |U|$, because every agent can possibly get value at most $|U|$. 
Adding up the values of agents over all stages, we obtain that the online algorithm returns a solution of expected value
\begin{equation*}
\sum_{j=1}^{t} V_j \leq \sum_{j=1}^{(1-\epsilon')t} \left(\frac{j}{t} + \epsilon'\right) n |U| + \epsilon' t n |U|
 \leq \frac{t}{2} n |U| + 2 \epsilon' t n |U| = \left(\frac12+2\epsilon'\right) t n |U|.
 \end{equation*}
In contrast, the offline optimum in the YES case is $t n |U|$ (by Claim~\ref{cl:opt}) and hence the $(\frac12+\delta)$-competitive algorithm must return expected value at least $(\frac12 + \delta) t n |U| = (\frac12 + 4 \epsilon') t n |U|$, a constant fraction better than the NO case. Since the possible values returned by the algorithm are in the range $[0, t n |U|]$, we can distinguish the two cases with constant two-sided error.

In fact, we can make the error one-sided as follows. If some agent receives $\ell \leq c_0 k$ items ($c_0$ as in the proof of Claim~\ref{cl:stage-j}) whose value is more than $(1 - (1-1/k)^\ell + \epsilon) |U|$, we answer YES, otherwise we answer NO. Note that by the proof of Claim~\ref{cl:stage-j}, in the YES case, we will answer YES with probability $\Omega(1)$, because otherwise the solution is almost always bounded by the same analysis as in the NO case, and the expected value of the solution would be less than $(\frac12 + 4 \epsilon') t n |U|$, which cannot be the case. In the NO case, we always answer NO, because there are no $\ell \leq c_0 k$ items of value more than $(1 - (1-1/k)^\ell + \epsilon) |U|$. Thus we can solve the Max $k$-cover decision problem with constant one-sided error, which implies $NP=RP$.
\end{proof}

\newcommand{\expect}{{\bf \mbox{\bf E}}}

\section{Online budget-additive allocation}
\label{sec:budget-additive}
In this section we prove that no online algorithm can obtain a better than $0.612$ approximation with respect to the standard LP in the budget-additive case. We now define the budgeted allocation problem\cite{CG08}.
\begin{definition}
Let $Q$ be a set of $m$ indivisible items and $A$ a set of $n$ agents, respectively,
where agent $a$ is willing to pay $b_{ai}$ for item $i$. Each agent $a$ has a budget constraint
$B_a$, and on receiving a set $S\subseteq Q$ of items pays $\min \{B_a, \sum_{i\in S} b_{ai} \}$. 
An allocation $\Gamma:A \to 2^{Q}$  is a partitioning of the items $Q$ into disjoint subsets $\Gamma(1),\ldots, \Gamma(n)$.
The maximum budgeted allocation problem, or simply MBA, is to find the allocation
which maximizes the total revenue, that is
$\sum_{a\in A}\min\left\{B_a,\sum_{i\in \Gamma(a)}b_{ai}\right\}$.
\end{definition}
Note that one can assume without loss of generality that $b_{ai}\leq B_a$, $\forall a\in A, i\in Q$.
Indeed, if bids are larger than budget, decreasing them to the budget does not
change the value of any allocation.
We now introduce the standard LP relaxation of the maximum budgeted  allocation problem \cite{CG08}:
\begin{equation}\label{eq:lp}
\begin{split}
& \max \sum_{a\in A, i\in Q} b_{ai} x_{ai}: \\
&\forall a\in A, \sum_{i\in Q} b_{ai}x_{ai}\leq B_a; \\
& \forall i\in Q, \sum_{a\in A} x_{ai}\leq 1;\\
&\forall a\in A, i\in Q, x_{ai}\geq 0.
\end{split}
\end{equation}

It was shown in \cite{CG08} that the integrality gap of LP \eqref{eq:lp} is exactly $3/4$. We now show that no online algorithm can obtain value better than a factor $0.612$ of this LP. Thus, if a $(1-1/e)$-competitive algorithm exists, it has to use other techniques, perhaps a stronger LP relaxation.

Our basic building block will be an instance with agents $A=\{a_1, a_2\}$ with budgets $B_{a_1}=B_{a_2}=3$ and items $I=\{i_1,i_2, i_3\}$ such that $b_{a_j, i_k}=2$ for all $j=1,2$ and $k=1, 2, 3$.  Note that the value of  the standard LP on this instance is $6$, while the maximum allocation is $5$ since an agent that gets two items can only pay $1$ for the second item that is allocated to him.
	
We now use the small instance that we just described to construct an online instance similarly to Section~\ref{sec:online-coverage}. Denote the set of agents in the system by $A$. We will have $|A|=2t$ and $B_a=3$ for all $a\in A$. Items will arrive in $t$ stages. It will be convenient to use a partition $A=\bigcup_{s=1}^t A^{(s)}$ of $A$ into disjoint sets of size $2$. The agents will gradually drop out of the system, i.e., agents who drop out at time $j$ will not be interested in items that arrive after $j$.  As before, for each $j=1,\ldots,t$ we refer to the set of agents that did not drop out before time $j$ as {\em active at time $j$}, and refer to the other sets of agents as {\em deactivated} at time $j$. Initially all sets $A^{(s)}, s=1,\ldots, t$ are active. After each stage, $A^{(s)}$ for $s$ uniformly random among the remaining active sets is deactivated. We denote the set deactivated after stage $j$ by $A_{j}$.

In each stage $j=1,\ldots, t$ a set of items $I_{j}$ arrives, where $|I_{j}|=3$ and $b_{ai}=2$ for all $a\in A$ that are active in stage $j$. Note that the value of the standard LP for our instance is $3t$: for each $j=1,\ldots, t$ allocate $2/3$ of each item in $I_j$ to each agent in the set $A_j$.
 
We now upper bound the value of any allocation that an online algorithm can obtain. 
Let 
\begin{equation}\label{def-g}
g(x)=\left\{
\begin{array}{cc}
2x,&\text{~if~}x\leq 1\\
x+1&\text{~if~}1\leq x\leq 2\\
3&\text{o.w.}
\end{array}
\right.
\end{equation}
 We first prove:
\begin{claim}\label{cl:g}
Let $a\in A$ denote an agent and let $X$ denote the (random) number of items allocated to $a$.
Then the expected value obtained by $a$ for these items is upper bounded by $g(\expect[X])$, where $g(\cdot)$ is given by \eqref{def-g}.
\end{claim}
\begin{proof}
This follows from concavity of $\min\left\{B_a, \sum_{i\in \Gamma(a)} b_{ai}\right\}$. Let $x=\expect[X]$. Then if $x<1$, the maximum is achieved if exactly one item is allocated to $a$ with probability $x$, yielding value $2x$. If $1\leq x\leq 2$, then 
the maximum is achieved if $1$ item is always allocated to $a$ at the price of $2$, and then a second item is allocated with probability $x$ at the price of $1$, yielding value $2+(x-1)=x+1$. Otherwise if $x\geq 2$, the payoff cannot be larger than the budget of $a$, i.e.~$3$.
\end{proof}

We can now prove:
\begin{theorem}
No online (randomized) algorithm for the budgeted allocation problem can achieve (in expectation) more than a $0.612$-fraction of the optimal value of the linear program (\ref{eq:lp}).
\end{theorem}
\begin{proof}
We first upper bound the expected number of items allocated to agents $A_j$ (recall that $A_j$ is the set of agents deactivated after stage $j$). Let $X^1_{j}, X^2_j$ denote the (random) number of items allocated to the two agents in $A_j$. By the same argument as in the proof of Claim~\ref{cl:bound}, which we do not repeat here, we have that each agent that is active at time $i=1,\ldots, j$ appears in $A_j$ with probability $\frac1{t-i+1}$. Since three items arrive in each stage, we have 
\begin{equation*}
\expect[X^1_j+X^2_j] \leq \sum_{i=1}^j\frac{3}{t-i+1}\leq 3\int_0^j\frac1{t-x}dx
=3\ln\left(\frac{t}{t-j}\right).
\end{equation*}

Now by Claim~\ref{cl:g} together with convexity of the function $g(\cdot)$ we get that 
the value obtained by the online algorithm is upper bounded by 
\begin{equation}\label{eq:ubound}
\sum_{j=1}^t \left[g(\expect[X^1_j])+g(\expect[X^2_j])\right]\leq \sum_{j=1}^{t} 2g\left(\frac32 \ln\left(\frac{t}{t-j}\right)\right)
\leq t\int_0^{1}2g\left(\frac32 \ln\left(\frac{1}{1-x}\right)\right)dx.
\end{equation}
We now split the interval $[0, 1]$ as $[0, 1]=[0, x_1]\cup [x_1, x_2]\cup [x_2, 1]$, where $\frac32 \ln\left(\frac{1}{1-x_1}\right)=1$ and $\frac32 \ln\left(\frac{1}{1-x_2}\right)=2$, i.e., $x_1=1-e^{-2/3}$ and $x_2=1-e^{4/3}$.
We get by \eqref{def-g} that the RHS of \eqref{eq:ubound} is equal to
\begin{equation*}
t\int_0^{1-e^{-2/3}}3 \ln\left(\frac1{1-x}\right)dx
+t\int_{1-e^{-2/3}}^{1-e^{-4/3}}\left[\frac32 \ln\left(\frac1{1-x}\right)+1\right]dx
+3t\cdot e^{-4/3} 
<  0.612\cdot 3t.
\end{equation*}
Recalling that the value of the standard LP on our instance is $3t$ completes the proof.
\end{proof}

\section{Stochastic allocation in the i.i.d.~model}
\label{sec:stochastic}

\paragraph{The i.i.d.~stochastic model.}
Here we consider a model where items arrive from some (possibly known or unknown) distribution $\cD$ over a fixed collection of items $M$. In each step, an item is drawn independently at random from $\cD$ and we must allocate it irrevocably to an agent. The total number of items can be either known or unknown.

In this model, we compare to the {\em expected offline optimum}, $OPT = \E[OPT(M)]$ where $M$ is the random multiset of items that appear on the input. We say that an algorithm is $c$-competitive if it achieves at least $c \cdot OPT$ in expectation over the random inputs (and possibly its own randomness).

\subsection{Diminishing returns on multisets}
\label{sec:diminishing}

In this section, we would like to consider the class of submodular valuations and its extension to multisets. Submodular valuations on $\{0,1\}^m$ express the property of diminishing returns, and this has indeed been the primary motivation for their modeling power as valuation functions. However, considering the stochastic setting with i.i.d.~samples, we should clarify how we deal with possible multiple copies of an item. In other words, we need to consider valuation functions $f:\ZZ_+^m \rightarrow \RR$.  An extension of submodularity to the $\ZZ_+^m$ lattice that has been used in the literature is the following condition: $f(x \vee y) + f(x \wedge y) \leq f(x) + f(y)$, where $\vee$ and $\wedge$ are the coordinate-wise max/min operations. Unfortunately, this condition does not quite capture the property of diminishing returns as it does in the case of $\{0,1\}^m$: note that in particular it does not impose any restrictions on $f(x)$ if the domain is $1$-dimensional, $x \in \ZZ_+$. Considering the property of diminishing returns, we would like the condition to imply that $f$ is concave in this $1$-dimensional setting. Therefore, we define the following property.

\begin{definition}
A function $f:\ZZ_+^m \rightarrow \RR$ has the property of diminishing returns, if for any $x \leq y$ (coordinate-wise) and any unit basis vector $e_i = (0,\ldots,0,1,0,\ldots,0)$, $ i \in [m]$,
$$ f(x+e_i) - f(x) \geq f(y+e_i) - f(y).$$
\end{definition}

Note that when restricted to $\{0,1\}^m$, this property is equivalent to submodularity.
Also, note that a simple way to extend a monotone submodular function $f:\{0,1\}^m \rightarrow \RR$ to $\tilde{f}: \ZZ_+^m \rightarrow \RR$, by declaring that additional copies of any item bring zero marginal value (i.e.~$\tilde{f}(x) = f(x \wedge \b1)$), satisfies the property of diminishing returns. In particular, coverage valuations on multisets interpreted in a natural way (multiple copies of the same set do not cover any new elements), have the property of diminishing returns. In some sense, we believe that this is the ``right extension" of submodularity to multisets, at least for applications related to combinatorial auctions and welfare maximization. 

We also consider the following natural notion of monotonicity.

\begin{definition}
A function $f:\ZZ_+^m \rightarrow \RR$ is monotone, if $f(x) \leq f(y)$ whenever $x \leq y$.
\end{definition}

\subsection{Our results}

We prove that in the i.i.d.~stochastic model with valuations satisfying the property of diminishing returns, the best one can achieve is a $(1-1/e)$-competitive algorithm. In fact, the factor of $1-1/e$ is achieved by the same greedy algorithm that gives a $1/2$-approximation in the adversarial online model \cite{FNW78,LLN06}.

\paragraph{Greedy algorithm:}
Suppose the multisets assigned to the $n$ agents before item $j$ arrives are $(T_1,\ldots,T_n)$.
Then assign item $j$ to the agent who maximizes $w_i(T_i+j) - w_i(T_i)$.

\

We remark that this algorithm obviously does not need to know the distribution or the number of items in advance.

\begin{theorem}
\label{thm:stochastic-greedy}
The greedy algorithm is $(1-1/e)$-competitive for welfare maximization with valuations satisfying the property of diminishing returns in the stochastic i.i.d.~model.
\end{theorem}

\begin{theorem}
\label{thm:stochastic-hardness}
Unless $NP = RP$, there is no $(1-1/e+\delta)$-competitive polynomial-time algorithm for welfare maximization with coverage valuations in the i.i.d.~stochastic model, for fixed $\delta>0$.
\end{theorem}

Since coverage valuations satisfy the property of diminishing returns, we conclude that $1-1/e$ is the optimal factor in the stochastic i.i.d.~model for coverage valuations as well as any valuations satisfying diminishing returns.

\subsection{Analysis of the greedy algorithm for stochastic input}

Here we prove Theorem~\ref{thm:stochastic-greedy}.
Our proof is a relatively straightforward extension of the analysis of \cite{DJSW11} in the budget-additive case. First, we need a bound on the expected optimum.

\begin{lemma}
\label{lem:LP-bound}
The expected optimum in the stochastic model where $m$ items arrive independently, item $j$ with probability $p_j$, is bounded by
\begin{eqnarray*}
LP & = & \max \sum_{i,S} x_{i,S} w_i(S): \\
& & \forall j; \sum_{i,S} x_{i,S} c_j(S) \leq p_j m; \\
& & \forall i; \sum_S x_{i,S} = 1; \\
& & \forall i,S; x_{i,S} \geq 0
\end{eqnarray*}
where $w_i$ is the valuation of agent $i$, $S$ runs over all multisets of at most $m$ items, and $c_j(S) \geq 0$ denotes the number of copies of $j$ contained in $S$.
\end{lemma}

\begin{proof}
Consider the optimal (offline) solution $OPT(M)$ for each realization of the random multiset $M$ of arriving items. Let $x_{i,S}$ denote the probability that the multiset allocated to agent $i$ in the optimal solution is $S$. Then the expected value of the optimum is $\E[OPT(M)] = \sum_{i,S} x_{i,S} w_i(S)$. Also, each multiset $S$ contains $c_j(S)$ copies of item $j$, so the expected number of allocated copies of item $j$ is $\sum_{i,S} x_{i,S} c_j(S)$. On the other hand, this cannot be more than the expected number of copies of $j$ in $M$, which is $\E[c_j(M)] = p_j m$. Therefore, $x_{i,S}$ is a feasible solution of value $OPT = \E[OPT(M)]$.
\end{proof}

\begin{lemma}
\label{lem:greedy-step}
Assume that $w_i$ are monotone valuations with the property of diminishing returns.
Condition on the partial allocation at some point being $(T_1,\ldots,T_n)$.
Then the expected gain from allocating the next random item is at least
$\frac{1}{m} (LP - \sum_i w_i(T_i))$.
\end{lemma}

\begin{proof}
Let $x_{i,S}$ be any feasible LP solution and let $y_{ij} = \sum_S x_{i,S} c_j(S)$. Recall that $c_j(S)$ denotes the number of copies of $j$ contained in $S$. Note that by the LP constraints, $\sum_i y_{ij} \leq p_j m$.
We use the notation $T_i+S$ to denote the union of multisets (adding up the multiplicities of each item). By the property of diminishing returns, we have
$$ w_i(T_i+S) - w_i(T_i) \leq \sum_j c_j(S) (w_i(T_i+j) - w_i(T_i)).$$
Adding up these inequalities multiplied by $x_{i,S} \geq 0$, we get
\begin{equation*}
\sum_{i,S} x_{i,S} (w_i(T_i+S) - w_i(T_i)) 
\leq \sum_{i,j,S} x_{i,S} c_j(S) (w_i(T_i+j) - w_i(T_i))
= \sum_{i,j} y_{ij}   (w_i(T_i+j) - w_i(T_i)).
\end{equation*}
Since $\sum_S x_{i,S} = 1$ by the LP constraints, and $w_i(T_i+S) \geq w_i(S)$ by monotonicity, we obtain
\begin{equation}
\label{eq:1}
\sum_{i,S} x_{i,S} w_i(S) - \sum_i w_i(T_i) \leq \sum_{i,j} y_{ij}  (w_i(T_i+j) - w_i(T_i)).
\end{equation}
Now consider a hypothetical allocation rule (depending on the fractional solution): If the incoming item is $j$, we allocate it to agent $i$ with probability $\frac{y_{ij}}{p_j m}$. (By the LP constraints, these probabilities for a fixed $j$ add up to at most $1$.) Since item $j$ appears with probability $p_j$, overall we allocate item $j$ to agent $i$ with probability $\frac{y_{ij}}{m}$. By (\ref{eq:1}), the expected gain of this randomized allocation rule is
\begin{equation*}
\E[\mbox{random gain}] = \sum_{i,j} \frac{y_{ij}}{m} (w_i(T_i+j) - w_i(T_i)) 
\geq \frac{1}{m}  \left(\sum_{i,S} x_{i,S} w_i(S) - \sum_i w_i(T_i) \right).
\end{equation*}
However, the greedy allocation rule gives each item to the agent maximizing her gain. Therefore, the greedy rule gains at least as much as the randomized allocation rule, for any feasible solution $x_{i,S}$.
This implies
\begin{equation*}
 \E[\mbox{greedy gain}] \geq \max \E[\mbox{random gain}] 
 \geq \frac{1}{m} \left(LP - \sum_i w_i(T_i)\right). \qedhere
 \end{equation*}
\end{proof}

Now we can prove Theorem~\ref{thm:stochastic-greedy}.

\begin{proof}[Proof of Theorem~\ref{thm:stochastic-greedy}]
Denote the allocation obtained after allocating $t$ items $(T^{(t)}_1,\ldots,T^{(t)}_n)$. Lemma~\ref{lem:greedy-step} states that conditioned on $(T^{(t)}_1,\ldots,T^{(t)}_n)$, the expected value after allocating 1 random item will be 
\begin{equation*}
 \E\Big[\sum_i w_i(T^{(t+1)}_i) \mid T^{(t)}_1,\ldots,T^{(t)}_n\Big] 
 \geq
 \sum_i w_i(T^{(t)}_i) + \frac{1}{m} \left(LP - \sum_i w_i(T^{(t)}_i)\right).
 \end{equation*}
Taking an expectation over the partial allocation $(T^{(t)}_1,\ldots,T^{(t)}_n)$, we obtain
\begin{equation*}
\E\Big[\sum_i w_i(T^{(t+1)}_i)\Big]
\geq
 \E\Big[\sum_i w_i(T^{(t)}_i)\Big] + \frac{1}{m} \E\Big[LP - \sum_i w_i(T^{(t)}_i)\Big].
 \end{equation*}
Let us denote $W(t) = \E[\sum_i w_i(T^{(t)}_i)]$. The last inequality states $W(t+1) \geq W(t) + \frac{1}{m} (LP - W(t))$, or equivalently $LP - W(t+1) \leq (1 - \frac{1}{m}) (LP - W(t))$. By induction, we obtain $$ LP - W(t) \leq \left(1-\frac{1}{m}\right)^t (LP - W(0)) \leq e^{-t/m} LP.$$
The expected value of the solution found by the greedy algorithm after $m$ items is $W(m) = \E[\sum_i w_i(T^{(m)}_i)]$; we conclude that $W(m) \geq (1 - 1/e) LP \geq (1-1/e) OPT$.
\end{proof}

\subsection{Optimality of $1-1/e$ in the stochastic i.i.d.~model}

Here we prove Theorem~\ref{thm:stochastic-hardness}.
We prove essentially that the stochastic online problem cannot be easier than the offline problem. However, the reduction is not quite black-box and we need some properties of the hard coverage instances that we discussed in Section~\ref{sec:online-coverage}.

\begin{proof}
Recall the instance $\cI$ of welfare maximization with coverage valuations (Section~\ref{sec:online-coverage}), for which it is NP-hard to achieve approximation better than $1-1/e$. We transform it into an instance $\cI^{[t]}$ in the stochastic i.i.d.~model as follows. We pick a parameter $t = poly(m)$ and produce $t$ identical copies of each agent in $\cI$. If the number of items in $\cI$ is $m$, we let $tm$ i.i.d.~items arrive from the uniform distribution on the $m$ items of $\cI$. By Chernoff bounds, with high probability the number of copies of each item on the input will be $t \pm O(\sqrt{t \log m}) = t \pm O(\sqrt{t \log t})$.

Consider the YES case. The items of $\cI$ can be allocated so that each of the $n$ agents covers the universe. Since we have at least $t - O(\sqrt{t \log t}) = (1-o(1)) t$ copies of each item with high probability, they can be allocated to $(1-o(1)) t n$ agents of the instance $\cI^{[t]}$ so that these agents get full value $|U|$. Thus the expected offline optimum is at least $(1 - o(1)) t n |U|$.

On the other hand, in the NO case, any agent in $\cI$ who gets $\ell \leq c_0 m / n$ items has value at most $(1 - (1 - n/m)^\ell + \epsilon) |U|$. Since the total number of items on the input of $\cI^{[t]}$ is $t m$ and the number of agents is $t n$, an agent can only get $m / n$ items on average. As the bound on the value as a function of the number of items is concave (and we can deal with the fact that this  bound works only up to $\ell \leq c_0 m / n$, similarly to Section~\ref{sec:online-coverage}), the optimum value is achieved if each agent receives $m/n$ items. Then the total value collected is $(1 - (1-n/m)^{m/n} + \epsilon) t n |U|$, which can be made arbitrarily close to $(1-1/e) t n |U|$. Note that this holds with probability $1$, irrespective of the randomness on the input.

If we had a $(1-1/e+\delta)$-competitive algorithm in the stochastic i.i.d.~model, we could distinguish these two cases with constant one-sided error, which would imply $NP = RP$.
\end{proof}

\pdfbookmark[1]{\refname}{My\refname}


\begin{thebibliography}{AGKM11}

\bibitem[AGKM11]{AGKM11}
G. Aggarwal, G. Goel, C. Karande, and A. Mehta.
\newblock Online vertex-weighted bipartite matching and single-bid budgeted allocations.
\newblock {\em Proc.~of ACM-SIAM SODA}, 2011.

\bibitem[BK10]{BK10}
B. Bahmani and M. Kapralov.
\newblock Improved bounds for online stochastic matching.
\newblock {\em Proc. of ESA}, 170-181, 2010.

\bibitem[BJN07]{BJN07}
N. Buchbinder, K. Jain, and S. Naor.
\newblock Online primal-dual algorithms for maximizing ad-auctions revenue.
\newblock {\em Proc. of ESA}, 253--264, 2007.

\bibitem[CG08]{CG08}
D. Chakrabarty and G. Goel.
\newblock On the approximability of budgeted allocations and improved
lower bounds for submodular welfare maximization and GAP.
\newblock {\em Proc.\ of IEEE FOCS}, 687--696, 2008.

\bibitem[DJSW11]{DJSW11}
N. Devanur, K. Jain, B. Sivan, and C. A. Wilkens.
\newblock Near optimal online algorithms and fast approximation algorithms for resource allocation problems.
\newblock {\em Proc. of ACM EC}, 2011.

\bibitem[DS06]{DS06}
S. Dobzinski and M. Schapira.
\newblock An improved approximation algorithm for combinatorial auctions
 with submodular bidders.
\newblock {\em Proc.\ of ACM-SIAM SODA}, 1064--1073, 2006.

\bibitem[Fei98]{Feige98}
U. Feige.
\newblock A threshold of $\ln n$ for approximating set cover.
\newblock {\em Journal of the ACM}, 45(4):634--652, 1998.

\bibitem[Fei06]{Fei06}
U. Feige.
\newblock On maximizing welfare when utility functions are subadditive.
\newblock {\em Proc. of ACM STOC}, 41--50, 2006.

\bibitem[FT04]{FT04}
U. Feige and P. Tetali.
\newblock Approximating min-sum set cover.
\newblock Algorithmica 40:4, 219--234, 2004.

\bibitem[FV10]{FV10}
U. Feige and J. Vondr\'ak.
\newblock The submodular welfare problem with demand queries.
\newblock {\em Theory of Computation} 6, 247--290, 2010.

\bibitem[FMMM09]{FMMM09}
J. Feldman, A. Mehta, V. Mirrokni and S. Muthukrishnan.
\newblock Online stochastic matching: Beating $1-1/e$.
\newblock {\em Proc. of IEEE FOCS}, 117--126, 2009.

\bibitem[FNW78]{FNW78}
M. L. Fisher, G. L. Nemhauser, and L. A. Wolsey.
\newblock An analysis of approximations for maximizing submodular set functions - II.
\newblock {\em Math. Prog. Study}, 8:73--87, 1978.

\bibitem[GKP01]{GKP01}
R. Garg, V. Kumar, and V. Pandit.
\newblock Approximation algorithms for budget-constrained auctions.
\newblock {\em Proc. of APPROX}, 102--113, 2001.

\bibitem[GM08]{GM08}
G. Goel and A. Mehta.
\newblock Online budgeted matching in random input models with applications to adwords.
\newblock {\em Proc. of ACM-SIAM SODA}, 982--991, 2008. 

\bibitem[KVV90]{KVV90}
R. Karp, U. Vazirani, and V. Vazirani.
\newblock An optimal algorithm for online bipartite matching.
\newblock {\em Proc.~of the $22^{nd}$ ACM STOC}, 1990.

\bibitem[KLMM08]{KLMM08}
S. Khot, R. Lipton, E. Markakis, and A. Mehta.
\newblock Inapproximability results for combinatorial auctions
with submodular utility functions.
\newblock {\em Algorithmica}, 52(1):3--18, 2008.

\bibitem[LLN06]{LLN06}
B. Lehmann, D. J. Lehmann, and N. Nisan.
\newblock Combinatorial auctions with decreasing marginal utilities.
\newblock {\em Games and Economic Behavior} 55:270--296, 2006.

\bibitem[MOS11]{MOS11}
V. Manshadi, S. Oveis Gharan, and A. Saberi.
\newblock Online stochastic matching: Online auctions based on offline statistics.
\newblock {\em Proc. ACM-SIAM SODA}, 1285--1294, 2011.

\bibitem[MSVV05]{MSVV05}
A. Mehta, A. Saberi, U. Vazirani, and V. Vazirani.
\newblock Adwords and generalized online matching.
\newblock {\em Proc.~of the $46^{th}$ IEEE FOCS}, 2005.

\bibitem[MOZ12]{MOZ12}
V. Mirrokni, S. Oveis Gharan and M. Zadimoghaddam.
\newblock Simultaneous approximations for adversarial and stochastic online budgeted allocation.
\newblock {\em Proc. of ACM-SIAM SODA}, 1690--1701, 2012.

\bibitem[NW78]{NW78}
G. L. Nemhauser and L. A. Wolsey.
\newblock Best algorithms for approximating the maximum of a submodular set function.
\newblock {\em Math. Op. Res.}, 3:177--188, 1978.

\bibitem[NWF78]{NWF78}
G. L. Nemhauser,  L. A. Wolsey, and M. L. Fisher.
\newblock An analysis of approximations for maximizing submodular set functions - I.
\newblock {\em Math. Prog.}, 14:265--294, 1978.

\bibitem[Von08]{V08}
J. Vondr\'ak.
\newblock Optimal approximation for the submodular welfare problem in the value oracle model.
\newblock {\em Proc. of the $40^{th}$ ACM STOC}, 67--74, 2008.

\end{thebibliography}
\end{document}